\documentclass[submission,copyright,creativecommons]{eptcs}
\providecommand{\event}{AFL 2026} 
\usepackage{iftex}
\ifpdf
  \usepackage{underscore}         
  \usepackage[T1]{fontenc}        
\else
  \usepackage{breakurl}           
\fi
\usepackage{amsmath,amssymb,multirow}
\usepackage{stmaryrd}
\usepackage{tikz}
\usepackage{url}
\usepackage{tikz}
\usetikzlibrary{automata,positioning,chains,fit,calc,arrows.meta,decorations.pathreplacing,matrix,arrows,shapes}
\usepackage{graphicx}
\usetikzlibrary{arrows.meta,positioning,automata}
\usepackage{graphicx} 
\newcommand\lB{\texttt{[}}
\newcommand\rB{\texttt{]}}

\usepackage{amsthm}
\newtheorem{theorem}{Theorem}
\newtheorem{lemma}[theorem]{Lemma}
\newtheorem{corollary}[theorem]{Corollary}

\theoremstyle{definition}
\newtheorem{definition}{Definition}

\renewcommand{\lB}{\mathopen{\scalebox{1.7}[1.0]{\texttt{[}}}}
\renewcommand{\rB}{\mathclose{\scalebox{1.7}[1.0]{\texttt{]}}}}
\DeclareRobustCommand{\lB}{\mathopen{\scalebox{1.7}[1.0]{\texttt{[}}}}
\DeclareRobustCommand{\rB}{\mathclose{\scalebox{1.7}[1.0]{\texttt{]}}}}
\newcommand{\lb}{\mathopen{\scalebox{.8}[1.0]{\texttt{\;[\;}}}}
\newcommand{\rb}{\mathclose{\scalebox{.8}[1.0]{\texttt{\;]\;}}}}
\newcommand{\lz}{\mathopen{\scalebox{.6}{$\texttt{[}^{0}$}}}
\newcommand{\rz}{\mathclose{\scalebox{.6}{$\texttt{]}^{0}$}}}
\newcommand{\ld}{\mathopen{\scalebox{1.4}[1.0]{\texttt{)}}}}
\newcommand{\rd}{\mathclose{\scalebox{1.4}[1.0]{\texttt{(}}}} 
\newcommand{\slB}{\mathopen{\!\!\scalebox{1.7}[1.0]{\texttt{[}}}}
\newcommand{\srB}{\mathclose{\scalebox{1.7}[1.0]{\texttt{]}}\!\!}}
\newcommand{\slb}{\mathopen{\scalebox{.8}[1.0]{\texttt{[}}}}
\newcommand{\srb}{\mathclose{\scalebox{.8}[1.0]{\texttt{]}}}}
\newcommand\quo[1]{\ \text`#1\text'\ }
\newcommand{\qlB}{\mathopen{\quo{\!\!\scalebox{1.7}[1.0]{\texttt{[}}}}}
\newcommand{\qrB}{\mathclose{\quo{\scalebox{1.7}[1.0]{\texttt{]}}\!\!}}}
\newcommand{\qlb}{\mathopen{\quo{\scalebox{.6}[1.0]{\texttt{[}}}}}
\newcommand{\qrb}{\mathclose{\quo{\scalebox{.6}[1.0]{\texttt{]}}}}}

\newcommand\stand{\mathrm{stand}}
\newcommand\lean{\mathrm{lean}}
\renewcommand{\int}{\mathrm{local}}
\newcommand{\local}{\mathrm{local}}

\newcommand{\Xrightarrow}[1]{\mathrel{\raisebox{-.7ex}{$\xrightarrow{#1}$}}}

\newcommand\semo[1]{\ensuremath{\llbracket #1 \rrbracket}}

\newcommand{\fold}{\ensuremath{\phi}} 
\newcommand{\unfold}{\ensuremath{\phi^{-1}}} 

\newcommand\zero{0}
\newcommand\rev{\operatorname{R}}

\usepackage{tikz}
\usetikzlibrary{automata,positioning,arrows.meta}

\tikzset{
  fsa/.style={
    >=Latex,
    node distance=12mm,
    every state/.style={
      circle,
      draw,
      minimum size=6mm,
      inner sep=0pt,
      font=\small
    },
    every edge/.style={
      draw,
      ->,
      shorten >=0pt,
      shorten <=0pt
    }
  },
  init/.style={initial, initial text={}},
  final/.style={accepting},
  edge/.style={->},
  lab/.style={font=\scriptsize, midway, above}
}

\title{Recognizers for Graph-Encoding Languages}
\def\titlerunning{Recognizers for Graph-Encoding Languages}
\author{Anssi Yli-Jyrä$^{1,2}$
 \institute{$^1$Digital Humanities, University of Helsinki, Finland}
 \institute{$^2$Mathematics Research Centre, Tampere University, Finland}
 \email{anssi.yli-jyra@helsinki.fi}}
\def\authorrunning{A. Yli-Jyr\"a}
\date{July 2026}

\begin{document}

\maketitle
\begin{abstract}
We introduce recurrent incidence automata (RIAs), a new automaton model
motivated by a decomposition of certain two-stack visibly pushdown
computations. The decomposition separates vertex-local finite-state
computations from recurrent one-stack interfaces connecting consecutive
vertices.
The construction is motivated by a two-stack visibly pushdown encoding of arbitrary
ordered graphs whose strings admit a unique factorization into
center-foldable vertex-local factors and whose auxiliary stack is empty at
every factor boundary. Folding each factor into a sequence of pair symbols
yields a local interface transformation.
An RIA consists of a finite-state unit that computes these transformations
and a recurrent layer that composes them across consecutive factors.
Rather than manipulating an internal pushdown store, RIAs externalize
long-range stack memory into recurrent interfaces between local
computations.
We show that nondeterministic RIA languages are closed under union,
intersection, concatenation, Kleene-$*$, and reversal. Deterministic RIAs
are 
closed under Boolean operations, although emptiness
remains undecidable.
\end{abstract}

\section{Introduction}

Visibly pushdown languages (VPLs) \cite{AlurMadhusudan2004,OKHOTIN201965} are among the most elegant language families due to their attractive closure and decidability properties, their ability to encode hierarchical structure, and their natural automaton model, visibly pushdown automata (VPAs). In parallel, a wide variety of encodings have been developed for sequence-to-sequence parsing and sequence labeling \cite{Diego-et-al-2026} in natural language processing.  Some encodings use balanced brackets to represent constituent structure \cite{NIPS2015_277281aa}. Others encode non-crossing graphs \cite{IJFCS,Ginter-2006,yli-jyra-gomez-rodriguez-2017-generic,DBLP:journals/jlm/Yli-Jyra19} or unrestricted ordered graphs \cite{Hierarchy,YliJyraFSMNLP2019,EzquerroEtAlACL2025}.
The latter encodings arise from the formal constraints of dependency structures \cite{G_mez_Rodr_guez_2026} and suggest connections with multi-stack visibly pushdown languages \cite{CAROTENUTO20161}. The encoding proposed by \cite{YliJyraFSMNLP2019} is especially relevant since it involves a visibly controlled 2-pushdown system and encodes arbitrary ordered graphs.

The key observation of this paper is that the graph-encoding computations of \cite{YliJyraFSMNLP2019} admit a decomposition into vertex-local finite-state transformations connected by recurrent one-stack interfaces. This decomposition motivates a new automaton model, recurrent incidence automata (RIAs), in which pushdown memory is externalized into recurrent interfaces between consecutive vertex-local factors.  The model separates local finite-state computation from long-range stack memory.


By isolating long-range memory into explicit interface transformations, RIAs expose a reusable algebraic representation of stack behavior that can be analyzed, composed, and constrained independently of the local computation.
This decomposition reflects the implicit restrictions of the underlying two-stack architecture and provides a natural basis for graph enumeration \cite{yli-jyra-gomez-rodriguez-2017-generic,Goyal}, neural network probing
\cite{ebrahimi-etal-2020-self,jerad2026contextfreerecognitiontransformers}, constrained structure decoding \cite{10.1007/978-3-540-93980-1_3,suresh2026dingo,EzquerroEtAlACL2025}, and related symbolic algorithms.

The design of RIAs is motivated by five characteristic structural properties of the graph encoding:

\begin{itemize}

\item \textbf{Vertex-local factorization.}  Each vertex corresponds to one center-foldable factor representing its incidences. 
In our encoding, this correspondence is realized through balanced VPL factors derived from Dyck primes \cite{Dyckprimes} rather than explicit delimiters \cite{Berstel_Perrin_Reutenauer_2009}. 
This preserves the balanced structure of each factor while allowing edge representations to span factor boundaries.
Vertex-local encoding also appears in the hardest CFL \cite{Greibach1973}, lexicalized TAGs \cite{JOSHI2004277}, dependency grammars and link grammars \cite{IJFCS,Ginter-2006}, sequence-labeling parsers \cite{ezquerro-etal-2024-dependency}, and contextual neural sequence encoders \cite{vaswani2017,devlin2019bert,hewitt2019structural}.

\item \textbf{Double-Balanced Bracketing.}
The encoding \cite{YliJyraFSMNLP2019} is characterized by two independent visibly pushdown structures defined over different visible partitions of the same alphabet.
Unlike multiplanar encodings \cite{Yli-Jyra2003,cao-etal-2017-parsing}, the additional pushdown is used only to access the first stack inside the local domains of vertices. 
The encoding is a homomorphic image of an intersection of two VPL languages \cite{YliJyraFSMNLP2019}, thereby suggesting that the language is a two-stack visibly pushdown language \cite{Bollig2008,CAROTENUTO20161}. 
This general characterization is more adequate without restrictions provided in \cite{LaTorreMadhusudanPa-ARobustClassofConte,CAROTENUTO20161}.

\item \textbf{Edge Partitioning.}  
In hierarchical bracketing, edges are partitioned into standalone and leaning edges, where each leaning edge is shorter than and shares an endpoint with a standalone edge
\cite{DBLP:journals/jlm/Yli-Jyra19,YliJyraFSMNLP2019,EzquerroEtAlACL2025}.
The graph's proper rope cover (PRC) uniquely determines this partition, is computed directly, and minimizes the number of standalone edges
\cite{YliJyraFSMNLP2019,EzquerroEtAlACL2025}.
Although the number of leaning edges may be unbounded, the information flow between the shared endpoint of a standalone edge and its locally leaning edges remains bounded.
This reduces the need to store edges separately in the stack.
Such endpoint locality has yielded measurable practical advantages
\cite{cao-etal-2017-parsing,DBLP:journals/jlm/Yli-Jyra19,EzquerroEtAlACL2025}.

\item \textbf{One-Stack Interface.} A crucial observation is that the auxiliary stack is always empty at the boundaries between consecutive vertex-local factors \cite{YliJyraFSMNLP2019}. Consequently, the main stack alone carries non-local information across these boundaries and naturally forms a recurrent one-stack interface. In the RIA model, a recurrent layer propagates the interface language through successive vertex-local computations \cite{Esparza,Uezato}. This decomposition resembles recurrence in language models \cite{hochreiter1997long}, rational recurrences \cite{peng-etal-2018-rational}, and graph-enumeration recurrences \cite{10.1007/978-3-031-89610-1_2}.

\item \textbf{Local Interface Transformations.}
The original encoding uses an auxiliary stack only within a vertex-local computation, leaving it empty at every factor boundary \cite{YliJyraFSMNLP2019}. We exploit this locality through \emph{center folding}, which compiles each local two-stack computation into a finite-state transformation of the persistent interface. Each factor thereby becomes a finite-state \emph{operator} that visibly rewrites the evolving incidence information carried by the interface. This operator view resonates with operator precedence grammars \cite{Floyd1963,CrespiReghizziMandrioli2012} and combinatory categorial grammars \cite{Steedman+2019+389+420}, where local operators compose to determine global structure. Associative composition of finite-state transformations combines the local computations into the global computation \cite{Berstel1979,Sakarovitch2009}.

\end{itemize}
Together, these observations motivate recurrent incidence automata (RIAs), a new automaton model based on recurrent interface transformations. Rather than manipulating a pushdown store throughout the computation, RIAs externalize long-range stack memory into a persistent one-stack interface and perform only finite-state computations locally. Unlike classical visibly pushdown automata, they introduce visible \emph{replace} operations, allowing local computations to update graph incidences directly, by rewriting interface symbols before recurrently composing the resulting interface transformations.

Our main theoretical results show that a two-stack visibly pushdown computation for graph-encoding languages can be decomposed into finite vertex-local interface transformations connected by a recurrent one-stack interface.
RIAs retain several attractive algebraic properties of visibly pushdown automata, and deterministic RIAs are closed under Boolean operations. 
We also define a subclass of RIAs that recognize languages of canonical graph-encodings.  The emptiness problem for RIAs is undecidable, as it is for two-stack visibly pushdown automata \cite{Carotenuto}.

The paper is organized as follows.
Section~2 describes the structure of graph-encoding strings.
Section~3 presents the RIA model and Section~4 its normalizations.
Section~5 studies the RIA-recognizable languages, and Section~6 studies decision problems. The paper is concluded by Section~7.

\paragraph{AI use.}
ChatGPT assisted with language, presentation, and discussion of proofs; all technical content and proofs were independently checked and verified by the author.

\section{Graph-Encoding Strings}
\label{sec:input-representation}

\paragraph{Double-Balanced Language.}

The input of a recurrent incidence automaton carries two independent
structures.  Locally, it is uniquely factorized into center-foldable
factors.  Globally, selected bracket symbols are balanced across factor
boundaries.  These structures are described by two visibly pushdown
languages defined over different visible partitions of the same alphabet.

A \emph{partitioned alphabet} is a tuple
$\Pi=
(\Sigma_{\stand},\Omega_{\stand},
 \Sigma_{\lean},\Omega_{\lean},
 \Sigma_{\local},\Omega_{\local})
$
of pairwise disjoint finite alphabets whose symbols are called
standalone brackets, leaning brackets, and local symbols,
respectively.
Define
$\widehat\Sigma=\Sigma_{\stand}\cup\Sigma_{\lean},\;
\widehat\Omega=\Omega_{\stand}\cup\Omega_{\lean},\;
\Sigma=\widehat\Sigma\cup\Sigma_{\local}$, and 
$\Omega=\widehat\Omega\cup\Omega_{\local}$.
Fix distinguished local symbols
$\quo\rd\in\Sigma_{\local}, \quo\ld\in\Omega_{\local}$.
Although written as brackets, they are local symbols and do not
participate in global bracket matching.

The language of \textit{center-foldable factors} is
$\nabla = \{xy\mid x\in\Sigma_{\local}\Sigma^*,\;y\in\Omega^*\Omega_{\local},\;|x|=|y|\}.$
The language $\nabla$ is a visibly pushdown language and prefix-free.
Consequently, every word in $\nabla^*$ admits a unique factorization $w=\tau_1\tau_2\cdots\tau_m,
\tau_i\in\nabla$ into center-foldable factors.

\textit{Global bracket balance} applies to a subset of the partitioned alphabet.
Let $h:(\Sigma\cup\Omega)^*\rightarrow(\widehat\Sigma\cup\widehat\Omega)^*$ erase all local symbols.
Let $\Delta$ be the language of balanced bracket strings over
$\widehat\Sigma\cup\widehat\Omega$, where every symbol of
$\widehat\Omega$ is an opening bracket and every symbol of
$\widehat\Sigma$ is a closing bracket.  Bracket types are ignored:
every opening bracket may match every closing bracket.
Define $\Delta'=h^{-1}(\Delta)$.
Both $\Delta$ and $\Delta'$ are visibly pushdown languages.

The input language combines two independent visibly pushdown structures over the same alphabet. Locally, every word admits a unique factorization into center-foldable factors. Globally, after erasing the local symbols, the remaining standalone and leaning brackets form a balanced bracket string. Their intersection is
\[
L_{\mathrm{max}}=\nabla^*\cap\Delta'.
\]
A synchronous two-stack visibly pushdown automaton recognizes
$L_{\mathrm{max}}$, using one stack for $\nabla^*$ and the other for
$\Delta'$.

\paragraph{Ordered Graph Encoding.}

The intended inputs of recurrent incidence automata are ordered graphs
$\mathcal G=(V,<,E)$, where $(V,<)$ is a finite linear order and
$E\subseteq V\times V$ satisfies $u<v$ for every edge $(u,v)\in E$.
The encoding used throughout the paper is illustrated by the following example.

Each graph is represented by a double-balanced word in $L_{\mathrm{max}}$.
Every vertex contributes one center-foldable factor, and these factors are concatenated in vertex order.
A local factorization of a word $w$ is the unique decomposition
$w=\tau_1\cdots\tau_n$
obtained by splitting $w$ into center-foldable factors $\tau_i\in\nabla$.
The word in Example~\ref{eq:factors} splits into eight
center-foldable factors:

\begin{gather}
\underbrace{\rd\zero\,\slB\ld}\quad
\underbrace{\rd\zero\,\rz\slb\slB\ld}\quad
\underbrace{\rd\zero\,\srb\lz\slB\ld}\quad
\underbrace{\rd\rz\rz\srB\underline\zero\lz\lz\ld}\quad
\underbrace{\rd\zero\srB\underline\zero\slB\ld}\quad
\underbrace{\rd\srB\underline\zero\ld}\quad
\underbrace{\rd\srB\underline\zero\ld}\quad
\underbrace{\rd\ld}\label{eq:factors}
\end{gather}

Projecting away the local symbols yields the globally balanced bracket
string in Example~\ref{eq:bal}:

\begin{gather}\renewcommand\phantom[1]{{\color{black!35}#1}}
\phantom{\rd\zero\,}\underbrace{\slB\phantom{\ld\quad\rd\zero\,}\rz}%
\overbrace{\slb\;\underbrace{\slB\phantom{\ld\quad\rd\zero\,}\srb}%
\underbrace{\lz\slB\phantom{\ld\quad\rd}\rz\rz}%
\srB}\phantom{\underline{0}}\overbrace{\lz\;\underbrace{\lz\phantom{\ld\quad\!\!\rd\zero}\srB}\phantom{\underline\zero}%
\underbrace{\slB\phantom{\ld\quad\!\!\rd}\srB}\phantom{\underline\zero}%
\phantom{\ld\quad\!\!\rd}\srB}\phantom{\underline\zero}%
\phantom{\ld\quad\rd}%
\phantom{\ld}\label{eq:bal}
\end{gather}

The eight factors correspond to the eight vertices of the ordered graph
in Example~\ref{eq:graph}.

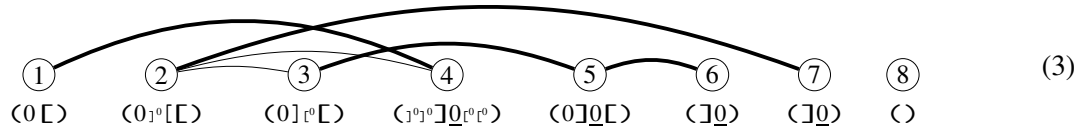
\begin{gather}
{\scalebox{.9}{\begin{tikzpicture}[baseline,
  vertex/.style={circle,draw,inner sep=1.5pt},
  edge/.style={-},
  ann/.style={font=\small, below=3mm},
  annx/.style={font=\scriptsize, below=16mm},
  foldedx/.style={font=\scriptsize, below=21mm},
  old/.style={font=\scriptsize, below=6mm},
  code/.style={font=\scriptsize, below=11mm},
  X folded/.style={font=\scriptsize, below=7mm},
  ]
\node[vertex] (v1) at (0,0) {$1$};
\node[vertex] (v2) at (1.8,0) {$2$};
\node[vertex] (v3) at (3.9,0) {$3$};
\node[vertex] (v4) at (6.0,0) {$4$};
\node[vertex] (v5) at (8.1,0) {$5$};
\node[vertex] (v6) at (9.9,0) {$6$};
\node[vertex] (v7) at (11.4,0) {$7$};
\node[vertex] (v8) at (12.7,0) {$8$};
\draw[edge,ultra thick,bend left=25] (v1) to node[above] {$ $} (v4);
\draw[edge,ultra thick,bend left=20] (v2) to node[above] {$ $} (v7);
\draw[edge,ultra thick,bend left=20] (v3) to node[above] {$ $} (v5);
\draw[edge,bend left=10] (v2) to node[below] {$ $} (v3);
\draw[edge,bend left=15] (v2) to node[below] {$ $} (v4);
\draw[edge,ultra thick,bend left=20] (v5) to node[below] {$ $} (v6);
\node[ann] at (v1) {$\rd\zero\,\slB\ld$};
\node[ann] at (v2) {$\rd\zero\,\rz\slb\slB\ld$};
\node[ann] at (v3) {$\rd\zero\,\srb\lz\slB\ld$};
\node[ann] at (v4) {$\rd\rz\rz\srB\underline\zero\lz\lz\ld$};
\node[ann] at (v5) {$\rd\zero\srB\underline\zero\slB\ld$};
\node[ann] at (v6) {$\rd\srB\underline\zero\ld$};
\node[ann] at (v7) {$\rd\srB\underline\zero\ld$};
\node[ann] at (v8) {$\rd\ld$};
\pgfresetboundingbox
\path[use as bounding box]
  (-0.35,-0.65) rectangle (13.05,-.65);
\end{tikzpicture}}}
\label{eq:graph}
\end{gather}

The alternating local and global bracket matchings induce the chains
shown in Example~\ref{eq:chains}.

\newcommand{\tok}[2]{\node[inner sep=1pt] (p#1) at (#1*.42,0) {$#2$};}
\newcommand{\majlink}[3]{%
  \draw[-{Latex[length=1.2mm]},blue!70!black,thick]
  (p#1.north) .. controls +(0,#3) and +(0,#3) .. (p#2.north);
}
\newcommand{\minlink}[3]{%
  \draw[-{Latex[length=1.2mm]},red!70!black,thick]
  (p#1.south) .. controls +(0,-#3) and +(0,-#3) .. (p#2.south);
}
\newcommand{\bminlink}[3]{%
  \draw[-{Latex[length=1.2mm]},blue!70!black,thick]
  (p#1.south) .. controls +(0,-#3) and +(0,-#3) .. (p#2.south);
}
\newcommand{\dotlink}[3]{%
  \draw[-{Latex[length=1.2mm]},red!70!black,thick]
  (p#1.south) .. controls +(0,-#3) and +(0,-#3) .. (p#2.south);
}
\newcommand{\loclink}[3]{%
}
\begin{gather}
\scalebox{.75}{\begin{tikzpicture}[baseline=(p1.base),font=\normalsize]
\tok{0}{\rd}\tok{1}0\tok{2}{\lB}\tok{3}{\ld}\tok{4}{\rd}
\tok{5}0\tok{6}{\rz}\tok{7}{\lb}\tok{8}{\lB}\tok{9}{\ld}\tok{10}{\rd}\tok{11}0
\tok{12}{\rb}\tok{13}{\lz}\tok{14}{\lB}\tok{15}{\ld}\tok{16}{\rd}
\tok{17}{\rz}\tok{18}{\rz}\tok{19}{\rB}\tok{20}{\underline{0}}\tok{21}{\lz}\tok{22}{\lz}\tok{23}{\ld}
\tok{24}{\rd}\tok{25}{0}\tok{26}{\rB}\tok{27}{\underline{0}}
\tok{28}{\lB}\tok{29}{\ld}\tok{30}{\rd}\tok{31}{\rB}\tok{32}{\underline{0}}
\tok{33}{\ld}\tok{34}{\rd}\tok{35}{\rB}\tok{36}{\underline{0}}
\tok{37}{\ld}\tok{38}{\rd}\tok{39}{\ld}
\dotlink{0}{3}{0.4}\minlink{1}{2}{0.2}
\majlink{2}{6}{0.4}\loclink{3}{4}{0.3}
\dotlink{4}{9}{0.5}\minlink{5}{8}{0.3}\bminlink{6}{7}{0.1}
\majlink{7}{19}{0.6}\majlink{8}{12}{0.3}\loclink{9}{10}{0.2}
\dotlink{10}{15}{0.5}\minlink{11}{14}{0.3}\bminlink{12}{13}{0.1}
\majlink{13}{18}{0.5}\majlink{14}{17}{0.3}\loclink{15}{16}{0.2}
\dotlink{16}{23}{0.56}\bminlink{17}{22}{0.4}\bminlink{18}{21}{0.3}\minlink{19}{20}{0.1}
\majlink{21}{35}{0.6}\majlink{22}{26}{0.3}\loclink{23}{24}{0.2} 
\dotlink{24}{29}{0.5}\minlink{25}{28}{0.3}\minlink{26}{27}{0.15}
\majlink{28}{31}{0.3}\loclink{29}{30}{0.2}
\dotlink{30}{33}{0.5}\minlink{31}{32}{0.2}
\loclink{33}{34}{0.3}
\dotlink{34}{37}{0.4}\minlink{35}{36}{0.2}
\loclink{37}{38}{0.2}
\dotlink{38}{39}{0.2}
\pgfresetboundingbox
\path[use as bounding box]
  (2,-0.65) rectangle (15.05,-.65);
\end{tikzpicture}}\label{eq:chains}
\end{gather}

The $\zero$- and $\rd\ld$-links in Example~\ref{eq:chains} are shown in red;
the remaining links are blue. Each blue chain connects a left bracket
$\qlB$ to a matching right bracket $\qrB$. Complete bracket chains
represent the \emph{standalone edges}, namely $(1,4)$, $(2,7)$, $(3,5)$, and
$(5,6)$ in the example.

The chains also encode \emph{leaning edges}, which share one endpoint
with a standalone edge. In the example, these are $(2,3)$ and $(2,4)$.
Their other endpoint is indicated by a bracket $\qlb$ or $\qrb$. The
resulting representation is called \emph{hierarchical bracketing}
\cite{EzquerroEtAlACL2025} because standalone edges form the proper rope cover, 
while the remaining edges are represented as leaning edges and their leaning brackets.

The correctness and compactness of the encoding rely on the notion of
\emph{proper rope cover} (PRC) \cite{YliJyraFSMNLP2019,EzquerroEtAlACL2025}.
A \emph{rope cover} is a subset $C\subseteq E$ such that every edge belongs to $C$ or leans on an edge of $C$.
A PRC is then a rope cover in which no two edges of $C$ lean on each other.  

An edge outside the PRC may have two supporting PRC edges, but the encoding 
of \cite{YliJyraFSMNLP2019} avoids redundant representations.
Thus, the standalone edges of the encoding are precisely the edges of the PRC, while every leaning edge is associated with one of its supporting PRC edges.
For example, the leaning edge $(2,4)$ could be represented using either
the standalone edge $(1,4)$ or the standalone edge $(2,7)$.  
Instead, only the shared right endpoint is used, as indicated by the leaning bracket $\qlb$ at vertex~2.

\paragraph{Center Folding.}

Each center-foldable factor $w=x_1\cdots x_ny_n\cdots y_1\in\nabla$
admits the canonical folding \allowbreak$\fold(w)=(x_1,y_1)\cdots(x_n,y_n)$,
which replaces each symmetric pair of symbols by a pair symbol.
Conversely, unfolding is defined recursively by 
$\unfold(\epsilon)=\epsilon$ and $\unfold((x,y)\alpha)=x\,\unfold(\alpha)\,y$.
For example, $\rd0\srB\srb\slb\underline0\slB\ld$ 
folds into the pair-symbol string $(\rd,\ld)\,(0,\slB)\,(\srB,\underline0)\,(\srb,\slb)$.

Center folding yields strings over the full pair alphabet
$\Sigma\times\Omega$.  The particular hierarchical encoding
considered here uses only pairs from
$(\Sigma_{\local}\times\Omega_{\stand})\cup
(\Sigma_{\stand}\times\Omega_{\local})\cup
(\Sigma_{\lean}\times\Omega_{\lean})\cup
(\Sigma_{\local}\times\Omega_{\local})$.

\paragraph{One-Stack Interfaces.}

At each boundary between consecutive factors, the active standalone edges form the one-stack interface.
The folded factor describes the vertex-local modifications to the interface.
A vertex may open or close a standalone edge, access an active edge 
through a leaning bracket, or leave the interface unchanged.

Let $\Gamma$ be the stack alphabet. A finite-state transducer processes each folded factor left-to-right by composing its pair symbols into an interface transformation:
\begin{gather}
\alpha\rho \longmapsto \beta\rho,\qquad \alpha,\beta,\rho\in\Gamma^*,   
\label{eq:transform}
\end{gather}
where the accessed prefix $\alpha$ is replaced by $\beta$, while the
residual suffix $\rho$ remains unchanged.
Example~\ref{eq:interface} illustrates these vertex-local interface transformations.
The blue arrows show how the one-stack interface propagates recurrently through the graph.
Vertically aligned pair symbols specify the transformation applied to the interface at each vertex.
\renewcommand{\tok}[4]{\node[inner sep=1pt] (p#3) at (#1*.42,#2*.5) {$#4$};}
\renewcommand{\majlink}[3]{%
  \draw[-{Latex[length=1.2mm]},blue!70!black,thick]
  (p#1.east) .. controls +(0,0) and +(0,0) .. (p#2.west);
}
\renewcommand{\minlink}[3]{%
  \draw[-{Latex[length=1.2mm]},red!70!black,thick]
  (p#1.east) .. controls +(0,0) and +(0,0) .. (p#2.west);
}
\begin{gather}
\raisebox{-3ex}{\scalebox{.75}{\begin{tikzpicture}[baseline=(p1.base),font=\normalsize]
\tok{1.5}{1}{0}{\rd}\tok{1.5}{0}{1}{0}\tok{3}{0}{2}{\slB}\tok{3}{1}{3}{\ld}
\tok{6.5}{2}{4}{\rd}\tok{6.5}{1}{5}0\tok{6.5}{0}{6}{\rz}\tok{8}{0}{7}{\lb}\tok{8}{1}{8}{\slB}\tok{8}{2}{9}{\ld}
\tok{12.5}{3}{10}{\rd}\tok{12.5}{2}{11}0\tok{12.5}{1}{12}{\rb}\tok{14}{1}{13}{\lz}\tok{14}{2}{14}{\slB}\tok{14}{3}{15}{\ld}
\tok{19.5}{3}{16}{\rd}\tok{19.5}{2}{17}{\rz}\tok{19.5}{1}{18}{\rz}\tok{19.5}{0}{19}{\srB}\tok{21}{0}{20}{\underline{0}}\tok{21}{1}{21}{\lz}\tok{21}{2}{22}{\lz}\tok{21}{3}{23}{\ld}
\tok{26.5}{4}{24}{\rd}\tok{26.5}{3}{25}{0}\tok{26.5}{2}{26}{\srB}\tok{28}{2}{27}{\underline{0}}\tok{28}{3}{28}{\slB}\tok{28}{4}{29}{\ld}
\tok{31.5}{4}{30}{\rd}\tok{31.5}{3}{31}{\srB}\tok{33}{3}{32}{\underline{0}}\tok{33}{4}{33}{\ld}
\tok{35.5}{2}{34}{\rd}\tok{35.5}{1}{35}{\srB}\tok{37}{1}{36}{\underline{0}}\tok{37}{2}{37}{\ld}
\tok{39.5}{2}{38}{\rd}\tok{41}{2}{39}{\ld}
\minlink{0}{3}{0.4}
\minlink{1}{2}{0.2}
\majlink{2}{6}{0.4}
\minlink{4}{9}{0.5}\minlink{5}{8}{0.3}\minlink{6}{7}{0.1}
\majlink{7}{19}{0.6}\majlink{8}{12}{0.3}
\minlink{10}{15}{0.5}\minlink{11}{14}{0.3}\minlink{12}{13}{0.1}
\majlink{13}{18}{0.5}\majlink{14}{17}{0.3}
\minlink{16}{23}{0.56}\minlink{17}{22}{0.4}\minlink{18}{21}{0.3}\minlink{19}{20}{0.1}
\majlink{21}{35}{0.6}\majlink{22}{26}{0.3}
\minlink{24}{29}{0.5}\minlink{25}{28}{0.3}\minlink{26}{27}{0.15}
\majlink{28}{31}{0.3}
\minlink{30}{33}{0.5}\minlink{31}{32}{0.2}
\minlink{34}{37}{0.4}\minlink{35}{36}{0.2}
\minlink{38}{39}{0.2}
\pgfresetboundingbox
\path[use as bounding box]
  (2,-0.65) rectangle (15.05,1.95);
\end{tikzpicture}}}\label{eq:interface}
\end{gather}


\section{Recurrent Incidence Automata}


Recurrent incidence automata extend the visibly pushdown paradigm by introducing visible replace operations. Call and return operations maintain the persistent one-stack interface, while replace operations rewrite interface symbols directly. The replace operations avoid explicit unwrapping and rewrapping of the intervening stack prefix.
The resulting interface transformations are computed locally and then composed globally.
Figure~1 summarizes the architecture. The left component computes a local interface transformation, while the right component recurrently composes these transformations along the sequence of vertex-local factors.

\paragraph{The RIA Unit.}

An RIA unit maps one interface to another while reading one folded factor.

An \emph{RIA unit} is a tuple
$A=(Q,I,F,\Pi,\Gamma,\delta)$, where
\begin{itemize}\setlength\itemsep{0em}
\item $Q$ is a finite set of \textit{local states}; 
\item $I\subseteq Q$ is the set of \textit{initial local states}, and $F\subseteq Q\backslash I$ is the set of \textit{final local states};
\item $\Pi$ is a \textit{partitioned alphabet} with symbol classes
  $\Sigma_{\stand},\Omega_{\stand},\Sigma_{\lean},\Omega_{\lean},\Sigma_{\local},\Omega_{\local}$;
\item $\Gamma$ is a finite \textit{stack alphabet}; and
\item $\delta$ is a \textit{transition relation}.
\end{itemize}
This determines the \textit{left alphabet} $\Sigma=\Sigma_{\stand}\uplus\Sigma_{\lean}\uplus\Sigma_{\local}$, 
the \textit{right alphabet} $\Omega=\Omega_{\stand}\uplus\Omega_{\lean}\uplus\Omega_{\local}$, and
the \textit{pair alphabet} $\Lambda=\Sigma \times \Omega$.  Let $\tilde I = Q\backslash I$ denote the non-initial states. Abbreviate 
\emph{stack fragments} as $\Gamma_\epsilon=\Gamma\cup\{\epsilon\}.$
The input symbols belong to the pair alphabet $\Lambda$.  Only the following symbol-pair classes 
occur in the transition system:
\begin{gather}
    \overbrace{(\Sigma_{\local}\times\Omega_{\stand})}^{\textrm{calls}}\;\cup\;
    \overbrace{(\Sigma_{\stand}\times\Omega_{\local})}^{\textrm{returns}}\;\cup\;
    \overbrace{(\Sigma_{\lean}\times\Omega_{\lean})}^{\textrm{replaces}}\;\cup\;
    \overbrace{(\Sigma_{\local}\times\Omega_{\local})}^{\textrm{internals}}
\end{gather}
An RIA unit transforms an input interface (a stack configuration) into an output interface according to the
instructions encoded by a folded input string, as illustrated in Figure~\ref{fig:RIA}.
Operationally, each transition reads one folded symbol pair and transforms the currently accessed interface position.
The symbol-pair classes correspond to the following transition types:
\begin{itemize}\setlength\itemsep{0em}
\item $\Sigma_{\local}\times\Omega_{\stand}$: \textit{call} transitions insert a new stack symbol into the interface;
\item $\Sigma_{\stand}\times\Omega_{\local}$: \textit{return} transitions remove a stack symbol from the interface;
\item $\Sigma_{\lean}\times\Omega_{\lean}$: \textit{replace} transitions rewrite the currently accessed interface symbol;
\item $\Sigma_{\local}\times\Omega_{\local}$: \textit{internal} transitions are used for internal annotation, but
\textit{initial} transitions also propagate the finite-state control of the layer between the interfaces.
\end{itemize}
Since the fourth symbol-pair class has two uses, the transition relation is partitioned into the following five classes:
\begin{gather}\begin{alignedat}{20}
\delta_1&\subseteq \tilde{I} \times&(\Sigma_\local  &\times\Omega_\stand)&\times&\{\epsilon\} &\times\;\;\Gamma  &\times \tilde{I}, & \qquad &
\delta_2&\subseteq \tilde{I} \times&(\Sigma_\stand  &\times\Omega_\local)&\times&\;\;\Gamma   &\times\{\epsilon\}&\times \tilde{I}, \\
\delta_3&\subseteq \tilde{I} \times&(\Sigma_\lean   &\times\Omega_\lean )&\times&\;\;\Gamma   &\times\;\;\Gamma  &\times \tilde{I}, \\
\delta_4&\subseteq \tilde{I} \times&(\Sigma_\local  &\times\Omega_\local)&\times&\{\epsilon\} &\times\{\epsilon\}&\times \tilde{I}, & &
\delta_5&\subseteq I            \times&(\Sigma_\local  &\times\Omega_\local)&\times&\;\;\Gamma   &\times\;\;\Gamma  &\times \tilde{I}.
\end{alignedat}\end{gather}
The control state of the recurrent layer is stored as a \textit{carrier symbol} at the top of the interface.
Initial transitions have a distinct pattern as they carry the layer's state from the input interface to the output interface.
Therefore, the role of the symbol pairs in initial transitions is different from that in non-initial transitions. 

A transition $(q,x,\gamma,\gamma',q')\in\delta$ reads the pair symbol $x\in\Lambda$, changes the state from $q$ to $q'$, consumes $\gamma\in\Gamma_\epsilon$ from the currently accessed input interface position, and emits $\gamma'\in\Gamma_\epsilon$ to the output interface.  Here $\epsilon$ denotes the absence of a consumed or emitted interface symbol.  The transition is also denoted 
\begin{gather}
q\Xrightarrow{x/\gamma/\gamma'}q'.
\end{gather}

A local computation of $A$ on a nonempty pair-symbol string $u=x_1\cdots x_n\in\Lambda^+$ defines a transformation of a nonempty input interface $\sigma\in \Gamma^+$ to an output interface $\sigma'\in \Gamma^+$.  It consists of states $q_0,\ldots,q_n$, transitions 
$(q_0\Xrightarrow{x_1/\gamma_1/\gamma_1'}q_1),$
$(q_1\Xrightarrow{x_2/\gamma_2/\gamma_2'}q_2), \dots,$
$(q_{n-1}\Xrightarrow{x_n/\gamma_n/\gamma_n'}q_n),$
and a residual suffix $\rho\in\Gamma^*$ such that
\[
q_0\in I, \qquad \sigma=\gamma_1\cdots\gamma_n\rho, \qquad \sigma'=\gamma_1'\cdots\gamma_n'\rho.
\]
The occurrences of $\epsilon$ are omitted from the concatenations $\gamma_1\cdots\gamma_n$ and $\gamma_1'\cdots\gamma_n'$,
and the residual suffix $\rho$ remains unaccessed and unchanged.  
The local computation is \emph{accepting} if $q_n\in F$.

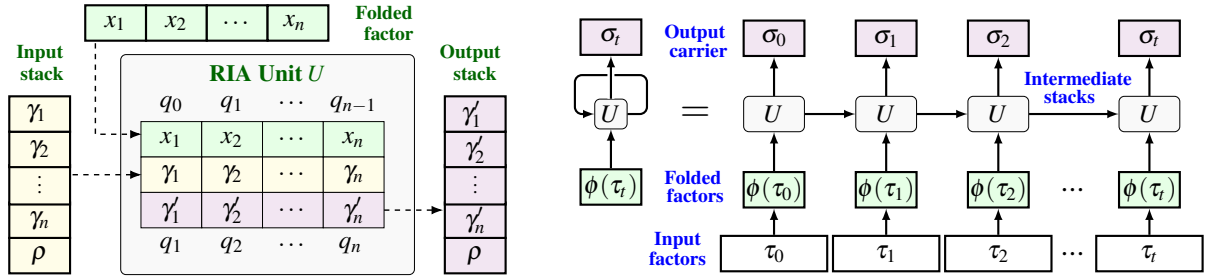
\begin{figure}\centering\scalebox{0.67}{
\begin{tikzpicture}[
    >=Latex,
    font=\normalsize,
    every node/.style={align=center},
    box/.style={draw, rounded corners=2pt, minimum height=7mm},
    tape/.style={draw, minimum width=10mm, inner sep=0pt, minimum height=7mm},
    dtape/.style={draw, minimum width=20mm, inner sep=0pt, minimum height=7mm},
    tape5/.style={draw, minimum width=13mm, inner sep=0pt, minimum height=7mm},
    tape3/.style={draw, minimum width=8mm, inner sep=0pt, minimum height=7mm},
    tape4/.style={draw, minimum width=7mm, inner sep=0pt, minimum height=7mm},
    tape2/.style={draw, minimum width=12mm, inner sep=0pt, minimum height=7mm},
    tape22/.style={draw, minimum width=26.5mm, inner sep=0pt, minimum height=7mm},
    tape20/.style={draw, minimum width=21mm, inner sep=0pt, minimum height=7mm},
    tape2i/.style={draw=none,fill=none, minimum width=15mm, inner sep=0pt, minimum height=7mm},
    stack/.style={draw, minimum width=12mm, inner sep=0pt, minimum height=10mm},
    stack2/.style={draw, minimum width=4mm, inner sep=0pt, minimum height=3mm},
    unit/.style={draw, rounded corners=3pt, fill=gray!5},
    factor1/.style={fill=green!15},
    factor2/.style={fill=yellow!20},
    factor3/.style={fill=violet!15},
    factor4/.style={fill=red!10},
    item/.style={circle,draw,thick,align=center},
    itemc/.style={item,on chain,join},
    unitde/.style={unit, minimum width=12mm, on chain, join},
    unitc/.style={draw, rounded corners=0pt, minimum width=12mm, minimum height=7mm, fill=red!5, on chain, join}
]
\begin{scope}
\node[font=\bfseries, text=green!40!black] at (-4.5,4.7) {Input};
\node[font=\bfseries, text=green!40!black] at (-4.45,4.3) {stack};
 \foreach \y/\lab in {
  1/{$\rho$}, 2/{$\gamma_n$}, 3/{\raisebox{1ex}{\scalebox{.8}{$\vdots$}}}, 4/{$\gamma_2$}, 5/{$\gamma_1$}
 }{\node[tape2, very thick, fill=yellow!10] (s\y) at (-4.5,\y*.7) {\scalebox{1.3}{\lab}};}
\node[font=\bfseries, text=green!40!black] at (4,4.7) {Output};
\node[font=\bfseries, text=green!40!black] at (4.1,4.3) {stack};
 \foreach \y/\lab in {
  1/{$\rho$}, 2/{$\gamma'_n$}, 3/{\raisebox{1ex}{\scalebox{.8}{$\vdots$}}}, 4/{$\gamma'_2$}, 5/{$\gamma'_1$}
 }{ \node[tape2, very thick, fill=violet!10] (t\y) at (4.1,\y*.7) {\scalebox{1.3}{\lab}}; }
\node[font=\bfseries, text=green!40!black] at (2.3,1+4.5) {Folded};
\node[font=\bfseries, text=green!40!black] at (2.4,1+4.1) {factor};
\foreach \x/\lab in {
  0/{$x_1$}, 1/{$x_2$}, 2/{$\cdots$}, 3/{$x_n$}
}{ \node[tape2, very thick, fill=green!10] (ft\x) at (\x*1.2-3,1+4.3) {\scalebox{1.3}{\lab}}; }
\node[unit, minimum width=5.8cm, minimum height=4.35cm] (U) at (0,1+1.5) {};
\node[font=\bfseries, text=green!40!black] at (0,4.3) {\scalebox{1.2}{RIA Unit $U$}};
\foreach \y/\a/\b/\c/\d/\e in {
 0/{$x_1$}/{$\gamma_1$}/{$\gamma_1'$}/{$q_0$}/{$q_1$}, 
 1/{$x_2$}/{$\gamma_2$}/{$\gamma_2'$}/{$q_1$}/{$q_2$}, 
 2/{\raisebox{0ex}{$\cdots$}}/{\raisebox{0ex}{$\cdots$}}/{\raisebox{0ex}{$\cdots$}}/{\raisebox{0ex}{$\cdots$}}/{\raisebox{0ex}{$\cdots$}},
 3/{$x_n$}/{$\gamma_n$}/{$\gamma_n'$}/{$q_{n-1}$}/{$q_{n}$}
}{
\node[tape2i]                at (\y*1.2-1.9,1.7+2)    {\scalebox{1.3}{\d}};
\node[tape2, fill=green!10]  at (\y*1.2-1.9,1.7+1.3)  {\scalebox{1.3}{\a}};  
\node[tape2, fill=yellow!10] at (\y*1.2-1.9,1.7+0.6)  {\scalebox{1.3}{\b}};
\node[tape2, fill=violet!10] at (\y*1.2-1.9,1.7-0.1)  {\scalebox{1.3}{\c}};
\node[tape2i]                at (\y*1.2-1.9,1.7-0.8)  {\scalebox{1.3}{\e}};
}
\draw[dashed,thick,->] (-3.85,2.3)--(-2.48,2.3);
\draw[dashed,thick,->] (2.3,1.6)--(3.47,1.6);
\draw[dashed,thick,->] (-3.4,1+3.9)--(-3.4,1+2.1)--(-2.48,1+2.1); 
\end{scope}


\begin{scope}[xshift=10cm,yshift=3.5cm,]
\begin{scope}[start chain=going right,nodes=unitc,every
 join/.style={-latex,very thick},local bounding box=chain]
 \path node[unit] (A0) {\scalebox{1.3}{$U$}} node[unit] (A1) {\scalebox{1.3}{$U$}} node[unit] (A2) {\scalebox{1.3}{$U$}} node[unit,xshift=2em] (At) {\scalebox{1.3}{$U$}};
 \end{scope}
 \node[left=1em of chain,scale=2] (eq) {$=$};
 \node[unit, left=2em of eq] (AL) {\scalebox{1.3}{$U$}};
 \path (AL.west) ++ (-1em,2em) coordinate (aux);
 \draw[very thick,-latex,rounded corners] (AL.east) -| ++ (1em,2em) -- (aux) 
 |- (AL.west);
 \foreach \X in {0,1,2,t} {
   \draw[very thick,unitc,-latex] (A\X.north) -- ++ (0,2em)
    node[above,tape2,fill=violet!10] (h\X) {\scalebox{1.3}{$\sigma_\X$}};
   \draw[very thick,unitc,latex-] (A\X.south) -- ++ (0,-2em)
    node[below,tape2,fill=green!10] (x\X) {\scalebox{1.3}{$\phi(\tau_\X)$}};
   \node[below=3ex of x\X,tape20,very thick] (tau\X) {\scalebox{1.3}{$\tau_\X$}};
   \draw[very thick,-latex] (tau\X.north) -- (x\X.south);
 }
 \draw[white,line width=0.8ex] (AL.north) -- ++ (0,1.9em);
 \draw[very thick,-latex] (AL.north) -- ++ (0,2.2em)
 node[above,tape2,fill=violet!10] {\scalebox{1.3}{$\sigma_t$}};
 \draw[very thick,latex-] (AL.south) -- ++ (0,-2em)
 node[below,tape2,fill=green!10] {\scalebox{1.3}{$\phi(\tau_t)$}};
 \path (x2) -- (xt) node[midway,scale=1,font=\bfseries] {\dots};
 \path (tau2) -- (taut) node[midway,scale=1,font=\bfseries] {\dots};

\node[font=\bfseries, blue, anchor=west] at (4.8,0.7) {Intermediate};
\node[font=\bfseries, blue, anchor=west] at (5.2,0.35) {stacks};
\node[font=\bfseries, blue, anchor=west] at (-2.3,1.6) {Output};
\node[font=\bfseries, blue, anchor=west] at (-2.2,1.25) {carrier};
\node[font=\bfseries, blue, anchor=west] at (-2.3,-1.25) {Folded};
\node[font=\bfseries, blue, anchor=west] at (-2.25,-1.6) {factors};
\node[font=\bfseries, blue, anchor=west] at (-2.5,-2.5) {Input};
\node[font=\bfseries, blue, anchor=west] at (-2.6,-2.85) {factors};

\end{scope}
\end{tikzpicture}
 
}\caption{The architecture of a recurrent incidence automaton.}
\label{fig:RIA}
\end{figure}

\paragraph{Stack-transformation relation.}

The relation $\semo A \subseteq \nabla \times \Gamma^* \times \Gamma^*$ is the \textit{stack-transformation relation} of unit $A$. 
\[
\semo A = \{ (w,\sigma,\sigma') \in \nabla \times \Gamma^* \times \Gamma^* \mid 
\textrm{$A$ has an accepting computation on $\fold(w) \in \Lambda^*$ from $\sigma$ to $\sigma'$}\;\}.
\]
Let $R,S\subseteq\nabla^*\times\Gamma^*\times\Gamma^*$ be stack-transformation relations.  
Their \textit{composition} is the relation 
\[
S\circ R = \bigl\{ (ww',\sigma,\sigma'') \mid \exists\sigma': (w,\sigma,\sigma')\in R \land (w',\sigma',\sigma'')\in S \bigr\}.
\]
The powers of a relation $R$ are defined recursively by
$R^0=\{(\epsilon,\sigma,\sigma)\mid \sigma\in\Gamma^*\}$
and $R^{n}=R^{n-1}\circ R \quad(n\ge1)$.
Thus, $R^n$ represents $n$ consecutive left-to-right applications of $R$.
The \emph{composition closure} of $R$ is $R^{\circledast}=\bigcup_{n\ge0}R^n$,
representing an arbitrary number of recurrent applications of $R$.

\paragraph{The RIA Layer.}

An RIA layer lifts the local stack transformations computed by a unit to
complete input strings through recurrent composition, as illustrated on the
right-hand side of Figure~\ref{fig:RIA}.

An \emph{RIA layer} is a tuple $\mathcal A=(A,G,H),$ where $A=(Q,I,F,\Pi,\Gamma,\delta)$ is an RIA unit,
and $G,H\subseteq\Gamma$ are the sets of \textit{initial and final carrier symbols}, respectively.
The language recognized by an RIA layer $\mathcal A=(A,G,H)$ is defined using the composition closure of the stack-transformation relation $\semo A$.
It is 
\begin{gather}
L_{\mathcal A} = \{\, w \mid (w,\sigma,\sigma') \in \semo A^{\circledast} \text{ for some } \sigma \in G^1 \text{ and } \sigma' \in H^1 \,\}.
\end{gather}
Languages recognized by RIA layers are RILs and their family is denoted by $\mathsf{RIL}$.

The \emph{maximal layer} $\mathcal U=(A,\{\iota\},\{\iota\})$ recognizes every RIA-admissible string $w\in\nabla^*\cap\Delta'$, i.e., every such string whose
folded factors use only the RIA-admissible symbol-pair classes.
Its unit is $A=(\{0,1\},\allowbreak\{0\},\allowbreak\{1\},\allowbreak\Pi,\allowbreak\{\iota\},\delta)$ where 
\begin{align}
\delta={}&
\{0\xrightarrow{x/\iota/\iota}1\mid x\in\Sigma_\local\times\Omega_\local\}
\cup\{1\xrightarrow{x/\iota/\iota}1\mid x\in\Sigma_\lean\times\Omega_\lean\}
\cup\{1\xrightarrow{x/\epsilon/\iota}1\mid x\in\Sigma_\local\times\Omega_\stand\}
\notag\\
&\cup\{1\xrightarrow{x/\iota/\epsilon}1\mid x\in\Sigma_\stand\times\Omega_\local\}
\cup\{1\xrightarrow{x/\epsilon/\epsilon}1\mid x\in\Sigma_\local\times\Omega_\local\}.
\end{align}

\section{Normalisations}

\begin{lemma}\label{lem:carr}
    Any RIL can be represented by a layer whose initial and final carriers never occur in stacks that contain other elements.
\end{lemma}     
\begin{proof}[Proof sketch]
  Introduce distinguished copies of the initial and final carrier symbols for singleton stacks and ordinary copies for their occurrences below other stack symbols. The unit propagates the distinction when symbols are pushed, popped, or exposed. Consequently, distinguished carrier symbols occur only as singleton stacks, while the induced interface transformations remain unchanged up to the renaming.
\end{proof}

The unit \(A\) is \emph{complete} if its transition relation $\delta$ defines a transition 
$q \xrightarrow{x/\gamma/\gamma'} q'$ on some $\gamma' \in \Gamma_\epsilon$ for every state $q\in Q$, 
for every symbol pair $x\in \Lambda$ and stack fragments $\gamma \in \Gamma_\epsilon$ that make an admissible 
combination given the partitioning of the transition relations.

\begin{lemma}\label{lem:completion}
   Every RIL has an equivalent complete RIA layer. If the original layer
   is deterministic, the complete layer can be chosen deterministic.
\end{lemma}
\begin{proof}[Proof sketch]
Introduce a fresh rejecting stack symbol $\dagger$. Extend the unit by
a disjoint deterministic completion component whose transition relation 
is defined on every compatible input--stack pair outside the original 
domain. Once entered, the completion component cannot be left. Every 
computation in the completion component produces a stack containing $\dagger$, and 
no stack containing $\dagger$ has a final carrier. The extension is 
therefore complete, preserves determinism, and recognizes the same language.
\end{proof}

The unit \(A\) is \emph{unambiguous} if every $(u,\sigma,\sigma')\in\semo A$ has a unique accepting computation and \emph{functional} if \(\semo A\) is a partial function $(\Sigma\cup\Omega)^*\times\Gamma^*\rightharpoonup\Gamma^*$.


A layer $\mathcal A=(A,G,H)$ is \emph{deterministic} if $G$ is a singleton and $A$ is functional,
and \emph{unambiguously deterministic} if $A$ is additionally unambiguous.
These notions are semantic: the underlying unit may still be
nondeterministic.  

\section{Closure Properties}

RIAs were motivated by graph encodings, but the automaton model is of independent interest. 
The model separates local finite-state computation from long-range stack memory.
We therefore first establish closure properties for general RILs before returning to graph-encoding languages as a distinguished subclass.
Denote the family of visibly pushdown languages (VPL) by $\mathsf{VPL}$,
and the family of 2-stack visibly pushdown languages (2MVPL) by $\mathsf{2MVPL}$.

\subsection{Unrestricted RILs}

\begin{definition}
Let $A=(Q_A,I_A,F_A,\Pi,\Gamma_A,\delta_A)$ and $B=(Q_B,\allowbreak I_B,\allowbreak F_B,\Pi,\allowbreak\Gamma_B,\allowbreak\delta_B)$ be RIA units over the same partitioned alphabet $\Pi$.  Their \emph{direct product} is the unit $A\times B=(Q_A\times Q_B,\allowbreak  I_A\times I_B,\allowbreak  F_A\times F_B,\allowbreak\Pi,\allowbreak
 \Gamma_A\times\Gamma_B,\allowbreak
 \delta_{A\times B}),$ where
\begin{align}
    & ((q_A,q_B),x,(\gamma_A,\gamma_B), (\gamma_A',\gamma_B'),(q_A',q_B'))  \in\delta_{A\times B}\notag\\
         &\qquad\Longleftrightarrow (q_A,x,\gamma_A,\gamma_A',q_A')\in\delta_A \ \text{and}\ (q_B,x,\gamma_B,\gamma_B',q_B')\in\delta_B .
\end{align}
\end{definition}

\begin{lemma}
Let $A$ and $B$ be RIA units over the same partitioned alphabet $\Pi$.
A computation of $A\times B$ exists exactly when the corresponding
computations of both $A$ and $B$ exist, and its projections are those
component computations.
\end{lemma}
\begin{proof}[Proof sketch]
Every product transition projects to one transition of each component.
Conversely, component transitions reading the same folded symbol pair
form a product transition. The claim follows inductively over the input.
\end{proof}

\begin{theorem}
    The RILs are closed under union, intersection, concatenation, Kleene-$\ast$, and reversal.
    \label{thm:reg}
\end{theorem}

\begin{proof}[Proof sketch]
Let $\mathcal A=(A,G_A,H_A)$ and $\mathcal B=(B,G_B,H_B)$ be layers with units $A=(Q_A,I_A,F_A,\Pi,\Gamma_A,\delta_A)$ and $B=(Q_B,\allowbreak I_B,\allowbreak F_B,\Pi,\allowbreak\Gamma_B,\delta_B)$ sharing $\Pi$.   Without loss of generality, assume that their state sets and stack alphabets are pairwise disjoint.  Also, assume that their initial and final carriers do not occur in stacks that contain other symbols.


\textbf{Union} $L_{\mathcal A} \cup L_{\mathcal B}$ is given by layer 
$\mathcal C=(C,G_A\cup G_B,H_A\cup H_B)$ and unit $C=(Q_A\cup Q_B,I_A\cup I_B,F_A \cup F_B, \Pi, \Gamma_A\cup\Gamma_B, \delta_A\cup \delta_B)$.

\textbf{Intersection} $L_{\mathcal A}\cap L_{\mathcal B}$ is given by layer $\mathcal C=(A\times B,G_A \times G_B,H_A\times H_B)$. 

\textbf{Concatenation} $L_{\mathcal A}L_{\mathcal B}$ is given by layer $\mathcal C=(C,G_A\cup G_C,H_B\cup H_C)$ and unit $C=(Q_A\cup Q_B, I_A \cup I_B, F_A \cup F_B, \Pi, \Gamma_A\cup\Gamma_B, \delta_A\cup \delta_B\cup \delta_C)$ where $\delta_C=\{(q,x,\gamma,\gamma'',q') \mid (q,x,\gamma,\gamma',q') \in \delta_A, \gamma' \in H_A, \allowbreak \gamma'' \in G_B \}$, and where $G_C$ and $H_C$ are empty by default, but $G_C=G_B$ if $\epsilon \in L_{\mathcal A}$ and $H_C=H_A$ if $\epsilon \in L_{\mathcal B}$.

\textbf{Kleene-$\ast$} $L_{\mathcal A}^\ast$ is given by layer $\mathcal C=(C,G_A\cup\{\iota\},H_A\cup\{\iota\})$ and unit $C=(Q_A, I_A, F_A, \Pi, \Gamma_A\cup \{\iota\}, \delta_A \cup \delta_C)$ where $\delta_C=\{(q,x,\gamma,\gamma'',q') \mid (q,x,\gamma,\gamma',q') \in \delta_A, \gamma' \in H_A, \gamma'' \in G_A \}$ and $\iota$ is a fresh symbol not in $\Gamma_A$. 

\textbf{Reversal} $L_{\mathcal{A}}^{\rev}$ is given by layer 
$\mathcal A^{\rev}=(A^{\rev},H_A,G_A)$ 
and unit $A^{\rev}=(Q_A,I_A,F_A,\Pi^{\rev},\Gamma_A,\delta_A^{\rev})$ where 
$\Pi^{\rev} = (\Omega_{\stand},\Sigma_{\stand}, \Omega_{\lean},\Sigma_{\lean}, \Omega_{\local},\Sigma_{\local})$, and 
$\delta_A^{\rev} = \{(q\Xrightarrow{(b,a)/\gamma'/\gamma} q') \mid  (q\Xrightarrow{(a,b)/\gamma/\gamma'} q') \in \delta_A\}$.  
In reversal, the alphabet partition is reversed accordingly, so that the roles of
the left and right components are exchanged.  
The initial and final states of the unit are not exchanged: the
reversed unit still processes a folded factor in the same operational
direction. Reversal instead reverses the order of layer interfaces.
\end{proof}

\begin{theorem}\label{thm:vpl-intersection}
RILs are closed under intersection with visibly pushdown languages
over a compatible partitioned alphabet.  If the RIL is deterministic, 
the intersection is deterministic as well.
\end{theorem}

\begin{proof}[Proof sketch]
Let $\mathcal A$ be a RIA layer. Recall that $L_{\mathcal A}\subseteq \nabla^*\cap\Delta'$ where $\nabla^*$ and $\Delta'$ are VPLs over compatible visible
partitions of the same alphabet.  Let $K\subseteq\Delta'$ be a VPL over a partitioned alphabet $\tilde{S}=(\Omega_\stand\cup\Omega_\lean, \Sigma_\stand \cup \Sigma_\lean, \Omega_\local \cup \Sigma_\local)$. 
We construct a RIA layer $\mathcal C=(C,G,H)$ recognizing $K\cap\nabla^*$.

Let $\mathcal B$ be a VPA recognizing $K$, let $T$ be its set of control states, and let $\Xi\subseteq ((\Sigma \cup \Omega) \times \Gamma \cup (\Sigma \cup \Omega))$ be the alphabet of transition labels in
its control graph. For every ordered pair $(s,t)\in T^2$, regard $s$ and $t$ as the initial and final states of the induced subgraph, respectively. The resulting path label language $L_{s,t}\subseteq\Xi^*$ is regular.
Let $h:\Xi^*\to(\Sigma\cup\Omega)^*$ map each transition label to its input symbol. Restrict $L_{s,t}$ to path labels whose input belongs to $\Sigma_{\local}\Sigma^*\Omega^*\Omega_{\local}$.
The center-foldable path labels are described by a linear context-free grammar whose recursive productions pair any $\Sigma$-transition label with any $\Omega$-transition label:
$S\rightarrow \xi S \xi' \mid \epsilon$, where $h(\xi)\in \Sigma$ and $h(\xi')\in \Omega$.
Intersecting this with the restricted language yields
linear language $L_{s,t}\cap h^{-1}(\nabla)$.  Next, an RIA unit is used to represent this language.

Most rules the linear grammar translate directly to RIA unit transitions,
but some rules have to be rejected due to the unit's partitioned transitions that 
do not allow all symbol pairs as input.
The initial transitions on $(\Sigma_\local\times\Omega_\local)$ additionally encode the states
$s,t$ as the input and output carrier symbols, respectively. Taking disjoint
copies of the units for all pairs $(s,t)$ gives the unit $C$.  
The initial and final states of $\mathcal B$ then determine the sets $G$ and $H$ needed in the layer. 
Hence $\mathcal C$ recognizes $K\cap\nabla^*$.
Finally, $L_{\mathcal C}\cap L_{\mathcal A}=(K\cap \nabla^*) \cap L_{\mathcal A}=K\cap (\nabla^* \cap L_{\mathcal A})= K\cap L_{\mathcal A}$,
because $L_{\mathcal{A}} \subseteq\nabla^*$.
Since the VPA $\mathcal{B}$ can be assumed to be deterministic, the intersection preserves determinism if $L_{\mathcal A}$ is deterministic.
\end{proof}

\begin{theorem}\label{thm:2vpl}
    $\mathsf{RIL}\subseteq \mathsf{2MVPL}$. For any RIL, we have an equivalent two-stack visibly pushdown automaton.
\end{theorem}
\begin{proof}[Proof sketch]
    Let $\mathcal{A}$ be a recurrent incidence automaton. We construct a two-stack visibly pushdown automaton $B$ from the operational semantics of RIAs. Each factor executes one complete local computation of the underlying two-stack visibly pushdown automaton, but using the auxiliary stack as a temporary storage in the way discussed in \cite{YliJyraFSMNLP2019}. The first stack treats $\widehat\Omega$ as calls and $\widehat\Sigma$ as returns, and the internal stack treats $\Sigma$ as calls and $\Omega$ as returns. Chaining these computations yields a two-stack visibly pushdown automaton whose first stack persists between factors and whose second stack is emptied after each factor. The factorization of the input guarantees that the local computations compose correctly.   The reduction is effective.
\end{proof}

\subsection{Deterministic RILs}

The family of deterministic RILs is denoted by $\mathsf{DRIL}$, and its languages are called DRILs.    

\begin{theorem}\label{thm:boolean}\label{thm:boo}
DRILs are closed under all Boolean operations.
\end{theorem}
\begin{proof}[Proof sketch]
Let $\mathcal A=(A,\allowbreak G_A,\allowbreak H_A)$ be a RIA layer and
$\mathcal B=(B,G_B,H_B)$ a deterministic RIA layer over the same partitioned
folded alphabet. By Lemmas~\ref{lem:completion} and~\ref{lem:carr}, assume that $\mathcal B$ is complete and that symbols
of \(H_B\) occur only in singleton stacks. Hence every input has a  
unique output stack under \(B\), and this stack is accepting exactly
when its carrier belongs to \(H_B\).

\textbf{Relative complement.} $L_{\mathcal A}\setminus L_{\mathcal B}$
is recognized by $(A\times B,\; G_A\times G_B,\; H_A\times(\Gamma_B\setminus H_B)).$
Layer $\mathcal B$ rejects precisely when the carrier of its output stack
lies outside $H_B$.  
The product preserves determinism if $\mathcal A$ is
deterministic. Thus DRILs are closed under intersection and
relative complement. Complement is relative to the maximal
layer, and union follows by De Morgan's law.
\end{proof}

\subsection{Unit-Acyclic RILs}
A layer is \emph{unit-acyclic} if the transition graph of its unit is acyclic.  
The family of unit-acyclic RILs (UARILs) is denoted by $\mathsf{UARIL}$.

\begin{theorem}\label{thm:vpl}
    $\mathsf{UARIL}\subseteq\mathsf{VPL}$.
\end{theorem}
\begin{proof}[Proof sketch]
The construction follows Theorem~\ref{thm:2vpl}, replacing the second stack by finite-state control. While the first stack checks global bracket balance, the control stores the output prefix $\beta$ of the interface transformation~\eqref{eq:transform} as it reads the input prefix $\alpha$. Since the unit is finite and acyclic, only finitely many local interface transformations can arise, so each fits in the VPA control.
\end{proof}

\begin{theorem}\label{thm:uaclos}
    UARILs are deterministic RILs, and they are closed under relative complement.
\end{theorem}
\begin{proof}[Proof sketch]
By Theorem~\ref{thm:vpl}, every unit-acyclic RIL is a VPL.
Determinizing its VPA and applying
Theorem~\ref{thm:vpl-intersection} with the maximal layer yields an
equivalent deterministic RIA.  For $L_{\mathcal A},L_{\mathcal B}\in\mathrm{UARIL}$, complement a
deterministic VPA for $L_{\mathcal B}$ and intersect it with the
unit-acyclic layer $\mathcal A$ using Theorem~\ref{thm:vpl-intersection}. 
The resulting layer recognizes
$L_{\mathcal A}\setminus L_{\mathcal B}$ and remains unit-acyclic,
because every transition of the constructed unit is synchronized with
a transition of the unit of $\mathcal A$. Hence any directed cycle in
the constructed unit would project to a directed cycle in the unit of
$\mathcal A$, contradicting its acyclicity.
\end{proof}

\subsection{Graph-Encoding RILs}
We return to the graph encodings that motivated RIAs by defining graph-encoding RILs.

For graph-encodings, we use a minimal alphabet, thereby eliminating symbol-level ambiguity.
Define the partitioned alphabet $\Pi=(\{\srB\}, \{\slB\}, \{\srb,\rz\}, \{\slb,\lz\}, \{\rd\}, \{\ld\})$.
A bijective mapping from a PRC partitioning to a bracketing is guaranteed in \cite{YliJyraFSMNLP2019} using 
three symmetry-breaking conventions:

\begin{itemize}\setlength\itemsep{0em}
\item \textsc{NoTail}.
The subfactor $\quo{\rz\lz}$ is forbidden, preventing
unnecessary stack access.

\item \textsc{WriteTop}.
The opening standalone bracket $\qlB$ is treated as a push operation: 
It adds a new edge to the top of the existing edges and can only occur as the last 
operation in the vertex.

\item \textsc{Earliest}.
An edge that can lean on either of two standalone edges is encoded with
respect to the shared right endpoint.
This removes redundancy between the two possible encodings.
\end{itemize}

We now define a deterministic recurrent incidence automaton 
$\mathcal S=(S,\{\iota\},\{\iota\})$. Its unit is $S=(Q,\{0\},F,\Pi,\{\iota\},\delta),$
where $Q=\{0,\dots,5\}$ and $F=\{1,2,4\}$.
The finite-state control enforces the three encoding conventions:

\begin{itemize}\setlength\itemsep{0em}
\item State \(0\) takes the mandatory initial transition on
      \((\rd,\ld)\), propagating the carrier symbol.
\item State \(1\) enforces \textsc{WriteTop}: a call on \((0,\slB)\)
      may occur only right after the initial transition.
\item States \(1\) -- \(3\) enforce \textsc{Earliest} by 
      collecting all incidences that can lean right to this vertex by a return on \((\srB,0)\).
\item State \(4\) admits the remaining leaning incidences.
\item States \(3\) and \(5\) are nonfinal, preventing a factor from
      terminating immediately after the redundant pattern
      \((\srb_{0},\slb_{0})\), as required by \textsc{NoTail}.
\end{itemize}
\if x
\begin{itemize}\setlength\itemsep{0em}
\item State~0 starts with an initial transition on $\quo{(\rd,\ld)}$ that passes the carrier symbol.

\item State~1 satisfies \textsc{WriteTop} via an early call transition on $\quo{(0,\slB)}$.

\item States~1--3 enforce \textsc{Earliest} by avoiding $\quo{(\srb,\lz)}$ and $\quo{(\srb,\slb)}$
before $\quo{(\srB,\underline0)}$.

\item State \(4\) admits the remaining leaning incidences.

\item Due to \textsc{NoTail}, nonfinal states~3 and~5 are must-visits after useless replace transitions on $\quo{(\rz,\lz)}$.

\end{itemize}
\fi 
The transition relation is
\begin{gather}
\scalebox{.8}{$\begin{alignedat}{5}
\delta=\{
&(0\Xrightarrow{(\rd,\ld)/\iota/\iota}1),
(1\Xrightarrow{(0,\slB)/\epsilon/\iota}2),
(\{1,2,3\}\Xrightarrow{(\rz,\slb)/\iota/\iota}2),
(\{1,2,3\}\Xrightarrow{(\rz,\lz)/\iota/\iota}3),
(\{1,2,3\}\Xrightarrow{(\srB,\underline0)/\iota/\epsilon}4),
\\
&
(\{1,\dots,5\}\Xrightarrow{\scalebox{.7}{$\begin{array}{c}
(\srb,\slb)/\iota/\iota\\
\end{array}$}}4),
(\{1,\dots,5\}\Xrightarrow{\scalebox{.7}{$\begin{array}{c}
(\srb,\lz)/\iota/\iota
\end{array}$}}4),
(\{4,5\}\Xrightarrow{(\rz,\slb)/\iota/\iota}4),
(\{4,5\}\Xrightarrow{(\rz,\lz)/\iota/\iota}5)
\}.
\end{alignedat}$}
\end{gather}

Standalone brackets encode the endpoints of standalone PRC edges, whereas leaning
brackets encode incidences of leaning edges.  The transition relation
admits exactly the local bracket patterns satisfying \textsc{WriteTop},
\textsc{Earliest}, and \textsc{NoTail}.

\begin{theorem}\label{thm:graph}
The automaton $\mathcal S$ recognizes exactly the canonical encodings of all ordered graphs, and no other strings.
\end{theorem}
\begin{proof}[Proof sketch]
Every ordered graph has a unique PRC, from which a number of hierarchical bracketing encodings are obtained in $\nabla^* \cap \Delta'$.  
The three conventions restrict the set of possible encodings to unique ones \cite{YliJyraFSMNLP2019} and 
are implemented by the RIA $\mathcal S$.   Conversely, 
the \textsc{Earliest} convention merely requires all available
right-leaning edges to be encoded at the shared right endpoint,
\textsc{NoTail} eliminates useless brackets, and \textsc{WriteTop}
essentially allocates an arbitrary stack position for every 
standalone edge that starts.  
The conventions are mutually compatible and retain a unique canonical encoding of every ordered graph.
\end{proof}

By Theorem~\ref{thm:graph}, the set of canonical encodings for ordered graphs can now be named as $L_{\mathrm{graph}}$.  
Any RIA-recognizable sublanguage of $L_{\mathrm{graph}}$ is called a \emph{graph-encoding RIL}, or \emph{GRIL}.  

\begin{corollary}
GRILs are closed under union, intersection, concatenation, and Kleene-$\ast$.
Deterministic GRILs are closed under complement.
\end{corollary}

\begin{proof}[Proof sketch]
Every GRIL has the form $L_{\mathrm{graph}}\cap L$, where $L$ is an RIL. 
Since $L_{\mathrm{graph}}=L_{\mathcal S}$ and $\mathcal S$ is deterministic, 
closure under union and intersection, as well as closure of deterministic
GRILs under complement relative to $L_{\mathrm{graph}}$, follow
from the corresponding closure properties of RILs and deterministic RILs. 
Since concatenation and Kleene-$\ast$ operate on complete code strings, each of which is 
a concatenation of vertex-local factors, these operations
juxtapose the component graphs disjointly without joining them.  This is an acceptable semantics for
the concatenation and Kleene-$\ast$ operations of graph-encoding languages.
\end{proof}


\section{Decision Problems and Family Inclusions}

We have already provided some inclusion results for language families.
The families of regular, context-free, context-sensitive, and recursively enumerable languages are denoted by $\mathsf{REG}$, $\mathsf{CFL}$, $\mathsf{CSL}$, and $\mathsf{RE}$, respectively.
The class $\mathsf{NSPACE}(n)$ consists of the decision problems solvable by
a nondeterministic Turing machine using $O(n)$ space.

\begin{corollary}[From Theorem \ref{thm:vpl}]\label{thm:uadec}
    Emptiness, containment, and equivalence are decidable for $\mathsf{UARIL}$.
\end{corollary}

\begin{theorem}
    $\mathsf{REG} \not\subseteq \mathsf{RIL}$ and $\mathsf{VPL}\not\subseteq\mathsf{RIL}$. \label{thm:nonreg}
\end{theorem}
\begin{proof}[Proof sketch]
    Every RIL is contained in $\nabla^*$.
    Let $a$ be a symbol.   Then $\{a\}$ belongs to $\mathsf{REG} \backslash \mathsf{RIL}$.     
    Since $\mathsf{REG} \subseteq \mathsf{VPL}$, it follows that $\mathsf{VPL} \not \subseteq \mathsf{RIL}$.
\end{proof}

\begin{theorem}
    $\mathsf{RIL} \not\subseteq \mathsf{CFL}$. \label{thm:noncfl}
\end{theorem}
\begin{proof}[Proof]
    Pick the partitioned alphabet $\Pi=(\{\rB_1,\rB_2\},\{\lB_1,\lB_2\},\{\rb,\rz\},\{\lb,\lz\},\{\rd,a,0\},\{\ld,\underline b,\underline c\})$ and define the factor languages
    $
    A=\{ \rd \;a\!\lB\!_1\; \ld\}, B=\{ \rd 0\; \rb^k \; \rB\!\!_1\underline b \; \lb^k\; \lB\!\,_2\ld \; \mid \; k\in\mathbb N\}$,\textrm{ and } $C=\{\rd\;\rB\!\!_2\underline c\;\ld\}.$
    
    It is straightforward to construct a RIA layer $\mathcal A$ whose unit recognizes factors from $A \cup B \cup C$ and uses its stack to match every $\lB_1$ with a later $\rB\!\!_1$ and every $\lB_2$ with a later $\rB\!\!_2$. 
    Let $h:(\Omega\cup\Sigma)^*\to\{a,\underline b,\underline c\}^*$ be the erasing homomorphism that fixes $a$, $\underline b$, and $\underline c$ and
    maps every other symbol to $\epsilon$, and let $R=a^*\underline b^*\underline c^*$.    
    The intersection with $R$ restricts the factors to the order $A^*B^*C^*$. Matching the brackets labelled $1$ forces the numbers of
    $A$- and $B$-factors to be equal, while matching those labelled $2$ forces the numbers of $B$- and $C$-factors to be equal, as shown in Example~\ref{eq:abc}.
    Hence
    $h(L(\mathcal A))\cap R = \{\;a^n\underline b^n\underline c^n\mid n\geq0\;\},$ which is not a context-free language.
    \begin{gather}\renewcommand\phantom[1]{{\color{black!27}#1}}
    \scalebox{.75}{$
    \phantom{\rd\;}a\phantom{\!}\lB\!_1\phantom{\;\ld  \rd\;}a\phantom{\!}\lB\!_1\phantom{\;\ld  \rd\;}a\phantom{\!}\lB\!_1\phantom{\;\ld  
    \rd0\;}\rb\rb\;\rB\!\!_1\underline{b}\phantom{\;}\lb\lb\;\lB\!\,_2\phantom{\ld  
    \rd0\;}\rb\rb\;\rB\!\!_1\underline{b}\phantom{\;}\lb\lb\;\lB\!\,_2\phantom{\ld  
    \rd0\;}\rb\rb\;\rB\!\!_1\underline{b}\phantom{\;}\lb\lb\;\lB\!\,_2\phantom{\ld  
    \rd\;}\rB\!\!_2\underline{c}\phantom{\;\ld\rd\;}\rB\!\!_2\underline{c}\phantom{\;\ld \rd\;}\rB\!\!_2\underline{c}\phantom{\;\ld}
    $}\label{eq:abc}
    \end{gather}
    
    Assume that $L(\mathcal A)$ is context-free. Context-free languages are closed under homomorphisms and intersection with regular languages, so
    $h(L(\mathcal A))\cap R$ would be context-free. This contradicts the fact that $\{a^n\underline b^n\underline c^n\mid n\geq0\}$ is not context-free. 
\end{proof}

\begin{theorem}\label{thm:nspace}
    $\mathsf{RIL}\subsetneq\mathsf{CSL}$.
\end{theorem}
\begin{proof}[Proof sketch]
Consider the membership problem.
Fix a RIA layer $\mathcal A$, and let $w$ be an unfolded input word of
length $n$. A nondeterministic algorithm guesses and verifies a
factorization of $w$ into center-foldable factors using $O(n)$ space.
It then simulates the successive applications of the RIA unit, storing
only the current factor, the current stack, and a constant amount of
finite-state control.
The current factor has length at most $n$. Since the unit is fixed and every transition changes the stack by at most one symbol, the stack height is $O(n)$.
Consequently, throughout the computation the stack has height $O(n)$. 
Thus each application of the RIA unit, and hence the entire membership 
computation, can be simulated in nondeterministic space $O(n)$. Therefore
$L(\mathcal A)\in\mathsf{NSPACE}(n)$.  Since $\mathsf{NSPACE}(n)=\mathsf{CSL}$, 
it follows that $\mathsf{RIL}\subseteq\mathsf{CSL}$. 
\nocite{KURODA1964207}
However, since $\mathsf{REG} \not\subseteq \mathsf{RIL}$, $\mathsf{RIL}\subsetneq\mathsf{CSL}$.
\end{proof}

\begin{theorem}\label{thm:emp}
The emptiness problem is undecidable for $\mathsf{DRIL}$.
\end{theorem}
\begin{proof}[Proof sketch]
The halting problem of deterministic two-counter machines (deterministic Minsky machines) reduces to the emptiness of deterministic RIA layers.
In our simulation of a two-counter machine, the machine configuration is encoded as a string stored on the stack.  
The initial input factor encodes an initial machine configuration.
The remaining factors encode its unique deterministic computation.  Each deterministic computation step of the machine is simulated as a functional stack transformation computed by the RIA unit.
On acceptance, the stack is reduced to a final carrier stack by a deterministic sequence of recurrent transformations.  An accepted input string witnesses non-emptiness.
Since the halting problem is undecidable for Minsky machines, emptiness for the RIA simulations of these machines must also be undecidable.
\end{proof}

\begin{corollary}\label{thm:undecid}
Emptiness, containment, and equivalence are undecidable for
$\mathsf{DRIL}$ and $\mathsf{RIL}$.
\end{corollary}

\begin{proof}[Proof sketch]
The empty language is recognized by a deterministic RIA layer. Hence, if
containment were decidable for $\mathsf{DRIL}$, emptiness would also be
decidable by testing whether $L\subseteq\emptyset$. Likewise, if
equivalence were decidable, emptiness would be decidable by testing
whether $L=\emptyset$. This contradicts Theorem~\ref{thm:emp}. Since
$\mathsf{DRIL}\subseteq\mathsf{RIL}$, undecidability of containment and
equivalence immediately extends to $\mathsf{RIL}$.
\end{proof}


\begin{table}[t]
\centering
\small
\scalebox{.76}{\begin{tabular}{l@{\;}l@{\;}l@{\;}l@{\;}l@{\;}l@{\;}l@{\;}l@{\;}l@{\;}l@{\;}l@{\;}l@{\;}l@{\;}l@{\;}l@{}l@{}l@{}l@{}l@{}l@{}l@{}l}
\hline
Family                   & Inclusion\ \  &$\cup$ &$\cap$ &$-$    & $L^R$\ \ \ \ \ 
{}                                                             &$\cdot,{}^*$\;&$\cap$REG\;&$\cap$VPL\;&$h$  & $h^{-1}$&Determinizable\;&Emptiness\;& $\subseteq$ & $=$\ \ \ \ \    & References\ \  & This paper \\ \hline
{}\sf REG                   &\sf VPL & yes   & yes   & yes   & yes   & yes          & yes   & yes    & yes & yes     & yes  & yes   & yes        & yes     &\cite{Hopcroft+Ullman/79/Introduction} \\
\sf UARIL& \sf VPL$\cap$DRIL\ \ \ 
{}                                   &\bf yes&\bf yes&\bf yes&\bf yes&\bf yes       &\bf yes&\bf(yes)&     &         & yes  & yes   & yes        & yes     &\cite{AlurMadhusudan2004} & \ref{thm:reg},\ref{thm:vpl},\ref{thm:uaclos},\ref{thm:uadec} \\ 
\sf VPL    &\sf FL,PB/OMVPL     & yes   & yes   & yes   & yes   & yes          & yes   &\bf(yes)&(yes)&(yes)    & yes  & yes   & yes        & yes     &\cite{AlurMadhusudan2004} \\
\sf FL               &\sf  DCFL\ \ \ & yes   & yes   & yes   & yes   & yes\ \ \ \   & yes   &        &     &         & yes  & yes   & yes        & yes     &\cite{CrespiReghizziMandrioli2012,doi:10.1137/140978818} \\ \hline
\sf DCFL                   &\sf  CFL & no    & no    & yes   & no    & no           & yes   &        & no  & yes     & immediate & yes   & no         & yes     &\cite{Hopcroft+Ullman/79/Introduction,Senizergues}\\\hline
\sf CFL                    &\sf  LIL & yes   & no    & no    & yes   & yes          & yes   &        & yes & yes     & no   & yes   & no         & no      &\cite{Hopcroft+Ullman/79/Introduction}\\
\sf LIL                    &\sf  IL  & yes   & no    & no    & yes   & yes          & yes   &        & yes & yes     & no   & yes   & no         & no      &\cite{LIL}\\
\sf IL                     &\sf  CSL & yes   & no    & no    & yes   & yes          & yes   &        & yes & yes     & no   & yes   & no         & no      &\cite{Aho}\\\hline
\sf PB-MVPL\ \         &\sf  $n$MVPL & yes   & yes   & yes   & no    & yes,no       & yes   &(yes)    &     &         & no   & yes   & yes        & yes     &\cite{Torre2013,LaTorreMadhusudanPa-ARobustClassofConte}\\
\sf OMVPL              &\sf  $n$MVPL & yes   & yes   & yes   & no    & no           & yes   &(yes)    &     &         & no   & yes   & yes        & yes     &\cite{Torre2013,CAROTENUTO20161}\\\hline
\sf DRIL            \ \ \ &\sf  RIL &\bf yes&\bf yes&\bf yes&        &              &\bf yes&\bf(yes) &       &           & immediate &\bf no & \bf no     & \bf no  &&\ref{thm:vpl-intersection},\ref{thm:boo},\ref{thm:undecid}\\
\sf RIL               &\sf  $n$MVPL &\bf yes&\bf yes&       &\bf yes& \bf yes       &\bf yes&\bf(yes) &       &           &           &\bf no & \bf no     & \bf no  &&\ref{thm:reg}, \ref{thm:vpl-intersection},\ref{thm:2vpl},\ref{thm:undecid}\\
\sf $n$MVPL                &\sf CSL & yes   & yes   & no    & yes   & yes           & yes   &        &     &         & no   & no    & no         & no      &\cite{Carotenuto,Bollig2008,CAROTENUTO20161}\\
\sf CSL                    &\sf RE  & yes   & yes   & yes   & yes\; & yes           & yes   &        &yes\;& yes\;   & no   & no    & no         & no      &\cite{Hopcroft+Ullman/79/Introduction,KURODA1964207}\\
\hline
\end{tabular}}
\caption{\label{tab:comp}Comparison of closure properties, determinization, decidability, and expressive power of related structured language families; (yes) denotes that the applicable operation has restrictions.}
\end{table}

\section{Conclusion}

RIAs suggest recurrent interface composition as a computational principle that separates local computation from long-range stack memory.
We have shown that a class of two-stack visibly pushdown computations admits a natural decomposition into vertex-local interface transformations connected by recurrent one-stack interfaces. Rather than maintaining multiple stacks throughout the computation, RIAs externalize persistent stack memory into recurrent interfaces while confining each local computation to a finite-state unit. This decomposition extends the locality of visibly pushdown computation and suggests recurrent composition of interface transformations as a general mechanism for computations over ordered objects. In particular, graph-encoding RILs provide an automata-theoretic representation of arbitrary ordered graphs through a bijection between valid encodings and graphs.

Table~\ref{tab:comp} highlights that RILs occupy an intermediate position between visibly pushdown and two-stack visibly pushdown languages.
The membership problem for $\mathsf{RIL}$ is decidable in nondeterministic linear space.
The language family has strong closure properties, including closure under intersection with compatible VPLs.
Deterministic RILs retain one of the hallmark properties of VPLs by remaining closed under all Boolean operations, and unit-acyclic RILs form a subclass of VPLs, inheriting many of their closure and decidability properties.
The emptiness problem for $\mathsf{RIL}$ is undecidable, as it is for two-stack VPLs.

Important research directions include identifying restrictions on interface transformations that recover determinizability and decidability of emptiness.
Further closure properties of deterministic RILs remain to be investigated, and the conditions under which RIAs are determinizable should be better understood.
Finally, extending RIAs with weights would support applications such as constrained decoding \cite{10.1007/978-3-540-93980-1_3,suresh2026dingo} 
and statistical parsing \cite{EzquerroEtAlACL2025}.

\bibliography{bib-brak,bib-nn,bib-span}
\end{document}